\documentclass[a4paper,10pt]{article}

\usepackage[top=2.5cm, bottom=2.5cm, left=2.5cm, right=2.5cm]{geometry}
\usepackage{amsmath,amssymb,amsfonts,amsthm,MnSymbol}
\usepackage{graphics,subfigure,color} 
\usepackage{latexsym}
\usepackage[dvips]{graphicx}
\usepackage{epsfig}
\usepackage{hyperref}
\usepackage{chngcntr}
\usepackage{psfrag}
\usepackage{multirow}
\usepackage{braket}
\usepackage{float}
\usepackage{enumerate}
\usepackage{cite}
\usepackage{empheq}

\newcommand{\bea}{\begin{eqnarray}}	
\newcommand{\eea}{\end{eqnarray}}
  
\DeclareMathOperator{\Tr}{Tr}

\newtheorem{proposition}{Proposition}[section]

\newtheorem{theorem}{Theorem}[section]

\theoremstyle{definition}

\allowdisplaybreaks[4]

\begin{document}

\title{The Wilson loop in the Gaussian Unitary Ensemble}
 
\author{
 \ Razvan Gurau\footnote{rgurau@cpht.polytechnique.fr, 
 Centre de Physique Th\'eorique, \'Ecole polytechnique, CNRS, Universit\'e Paris-Saclay, F-91128 Palaiseau, France
 and Perimeter Institute for Theoretical Physics, 31 Caroline St. N, N2L 2Y5, Waterloo, ON, Canada.}
} 

\maketitle

\begin{abstract} 
Using the supersymmetric formalism we compute exactly at finite $N$ the expectation of the Wilson loop in the Gaussian Unitary Ensemble 
and derive an exact formula for the spectral density at finite $N$. 
We obtain the same result by a second method relying on enumerative combinatorics and show that it 
leads to a novel proof of the Harer-Zagier series formula. 
\end{abstract}


\section{Introduction}

The Gaussian Unitary Ensemble (GUE) is among the best known classical random matrix ensembles \cite{Gui}.
Like in other random matrix ensembles, all the interesting quantities in the GUE
(\emph{e.g.} the eigenvalue density) have a controlled $1/N$ expansion, where $N$ is the size of the matrix.
This expansion can be systematically explored order by order:  in particular the eigenvalue density
converges in the large $N$ limit to the Wigner semicircle law  \cite{Wigner}.

The GUE has been studied by a variety of methods, from orthogonal polynomials and combinatorial techniques  \cite{DiFrancesco:1993nw} to the 
supersymmetric formalism \cite{efetov1983supersymmetry,disertori2003random}. One of the most striking features of 
the supersymmetric formalism is that one is able to find integral expressions (for instance for the spectral density)
which depend only parametrically on $N$ and are adapted to a saddle point analysis. The $1/N$ expansion 
is systematically recovered by computing the corrections to the leading saddle point result \cite{shamis2013density}.

In this paper we do not use this standard saddle point analysis, but rather use 
the supersymmetric formalism to derive exact results at finite $N$ in the GUE. We first obtain in Section \ref{sec:WL}
an exact expression for the expectation of the Wilson loop observable. Using this result,
we obtain in Section \ref{sec:spectden} an exact expression for the spectral density in the GUE at finite $N$ and discuss
its $1/N$ expansion. In the second part of this paper, section \ref{sec:combi}, we obtain the expectation of the Wilson loop by a second method
using a version of the BEST theorem \cite{aardenne1951circuits,tutte1941unicursal} in combinatorics adapted to undirected multi graphs
and find a novel proof of the Harer-Zagier series formula \cite{DBLP:journals/jct/ChapuyFF13}.

\section{The Wilson loop in the GUE}\label{sec:WL}

Let $H$ be a $N\times N$ Hermitian matrix and let us denote:
\[
 [ dH ] \equiv \prod_{a} \frac{ \sqrt{N} dH_{aa}}{\sqrt{2\pi }} \prod_{a<b} \frac{N dH_{ab}d\bar H_{ab}}{ 2 \pi \imath} \; , \qquad 
 \Braket{f(H)}_{GUE} \equiv \int [dH] \; e^{-\frac{N}{2}\Tr[H^2] } \; f(H) \;,
\]
with $f$ any function of $H$.
The expectation of the Wilson loop observable in the GUE is:
\begin{equation}\label{eq:WilsonLoop}
 I(t,N) \equiv \Braket{\frac{1}{N} \Tr\left[ e^{\imath t H}\right] }_{GUE} =  \int [dH] \; e^{-\frac{N}{2} \Tr[H^2]} \; \frac{1}{N} \Tr\left[ e^{\imath t H}\right] \;,
\end{equation}
where $t$ is a complex number. The main result of this paper is an exact formula for this expectation.
\begin{theorem}\label{thm:WL}
 The expectation of the Wilson loop in the GUE is:
 \[
\boxed{   I(t,N)
  = e^{-\frac{t^2}{2N}}\left[ \sum_{q=0}^{N-1}  \frac{1}{N^{q+1}} \binom{N}{q+1} \frac{ \left(-t^2 \right)^q } {q!}\right]  \;. }
 \]
\end{theorem}
\begin{proof}
 See section \ref{ssec:proof}
\end{proof}

This result has already been derived in the literature \cite{Drukker:2000rr} using orthogonal polynomials. We derive it here using the supersymmetric formalism.
This exact formula has several interesting limit cases $  I(0,N)  = 1 , \; I(t,1) = e^{-\frac{t^2}{2}} $ and:
\[
 I(t,\infty)  = \sum_{q=0}^{\infty} \frac{(- t^2)^q }{(q+1)! q! }  \;.
\]
A tight upper bound on the expectation of the Wilson loop is:
 \begin{align*}
   |I( t ,N)| & \le   \left| e^{-\frac{t^2}{2N}} \right|  \left( \sum_{q=0}^{\infty} \frac{|t|^{2q}}{ (q+1)! q! }  \right) = e^{-\frac{1}{2N} \Re(t^2) }  
  \sum_{q\ge 0} \frac{1}{(2q)!} \frac{(2q)!}{ (q+1)! q! } |t|^{2q}    \le e^{-\frac{1}{2N} \Re(t^2) } e^{2|t|} \; , \;\; \forall t\in \mathbb{C} \;.
\end{align*}

\subsection{Proof of Theorem \ref{thm:WL}}\label{ssec:proof}

We first review briefly supersymmetric Gaussian integrals.
Following \cite{mirlin2000statistics,disertori2003random}, let us consider $2N$ Grassmann variables $\chi_i,\bar \chi_i\;, i=1,\dots N$ with:
\begin{align*}
& \chi_i \chi_j = -\chi_j \chi_i \;,  \chi_i \bar \chi_j = - \bar \chi_j \chi_i  \;,  \bar \chi_i \bar \chi_j = - \bar \chi_j \bar \chi_i \;, \crcr
& \int d\chi_i \;\chi_i = \frac{\partial }{\partial \chi_i} \chi_i = 1= \int d\bar \chi_i \; \bar \chi_i = \frac{\partial }{\partial \bar \chi_i} \bar \chi_i  \;,
\end{align*}
and $2N$ commuting variables $\bar \phi_i,\phi_i\; i=1\dots N$. Denoting $[d\bar \chi d\chi] \equiv   \prod_{i=1}^N  d\bar \chi_i d\chi_i    $ and 
$  [d\bar \phi d\phi] \equiv   \prod_{i=1}^N  \frac{ d \bar \phi_i  d\phi_i }{2\pi } $,
we have for any (invertible) $N\times N$ matrix $M$:
\begin{align*}
&  \int [d\bar \phi d\phi]  \;  e^{-\bar \phi M \phi}  =\frac{1}{\det(M)}\;, 
\qquad  \int  [d\bar \chi d\chi] \;  e^{ -\bar \chi M \chi}   = \det(M) \;, \crcr 
& \int  [d\bar \chi d\chi]  [d\bar \phi d\phi]   \;  e^{-\bar \phi M \phi -\bar \chi M \chi}  \; \bar \phi_a \phi_b = (M^{-1})_{ba} \;. 
\end{align*}

For any matrices $A,D$ with commuting entries and $B,C$ with Grassmann entries, the 
supersymmetric Gaussian integral:
\begin{align*}
 \int [d\bar \chi d\chi] [d\bar \phi d\phi] \;\; e^{- \begin{pmatrix}
                                                 \bar \phi & \bar \chi
                                                \end{pmatrix} 
                                                 \begin{pmatrix}
                                                  A & B \\ C & D 
                                                 \end{pmatrix}
                                                 \begin{pmatrix}
                                                    \phi \\ \chi 
                                                 \end{pmatrix}
                                                 } = \int [d\bar \chi d\chi] [d\bar \phi d\phi]  e^{-\bar \phi A \phi -  \left( \bar \chi + \bar \phi B D^{-1} \right) D \left( \chi   + D^{-1} C \phi  \right) + \bar \phi B D^{-1} C \phi} \;,
\end{align*}
is computed by changing variables to $\bar \psi = \bar \chi + \bar \phi B D^{-1}  $ and $  \psi =   \chi +  D^{-1} C \phi $ and equals the inverse of the Berezinian \cite{mirlin2000statistics,disertori2003random}:
\[
 \int [d\bar \chi d\chi] [d\bar \phi d\phi] \;\; e^{- \begin{pmatrix}
                                                 \bar \phi & \bar \chi
                                                \end{pmatrix} 
                                                 \begin{pmatrix}
                                                  A & B \\ C & D 
                                                 \end{pmatrix}
                                                 \begin{pmatrix}
                                                    \phi \\ \chi 
                                                 \end{pmatrix}
                                                 } =  \frac{\det D}{ \det (A - BD^{-1}C)} \;.
\]
In this formula $ \left[ \det (A - BD^{-1}C)\right]^{-1}$ is understood as a polynomial in the entries of $ B $ and $C$ starting
with a $ \left[ \det(A)\right]^{-1} $ term.

\paragraph{Integral representation of the resolvent.}
The expectation of the resolvent
can be seen as the Laplace transform of the expectation of the Wilson loop:
\[
 \omega_N(z) = \frac{1}{N} \Braket{ \frac{1}{z-\imath H} }_{GUE} = \int_0^{\infty} dt \; e^{-zt} I(t,N) \;,
\]
where the Laplace transform converges for $\Re(z)>0$. 
Observe that the resolvent as defined here differs by an overall constant factor from the standard definition \cite{DiFrancesco:1993nw}.
The Wilson loop is then the inverse Laplace transform of the resolvent:
\[
  I(t,N) = \int_{\gamma -\imath \infty}^{\gamma+\imath \infty} \frac{ds}{2\pi \imath}   \; e^{st} \omega_N(s) \;,
\]
where $\gamma$ is a positive real number. 

We represent the resolvent as a supersymmetric integral \cite{spencer2012susy,disertori2003random}:
\begin{align*}
 \omega_N(z) = \frac{1}{N}\Braket{ \Tr\left[ \frac{1}{z-\imath   H } \right] }_{GUE}  
 =   \int   [d\bar \chi d\chi]  [d \bar \phi  d\phi] \; \frac{\left( \bar \phi \cdot \phi \right)}{N}     e^{  - z \bar \phi \cdot \phi - z \bar\chi \cdot \chi   } 
  \;   \Braket{ e^{ \imath   \sum_{a,b}( \bar \phi_a \phi_b + \bar\chi_a \chi_b )  H_{ab} } }_{GUE} \;,
\end{align*}
where $\bar \phi \cdot \phi \equiv \sum_{i=1}^N \bar \phi_i \phi_i $, $ \bar\chi \cdot \chi   \equiv \sum_{i=1}^N \bar \chi_i \chi_i$, and the integrals 
over the commuting variables converge as $\Re(z) > 0$. Observe that for any fixed matrix $M$ we have:
\[
 \Braket{ e^{ \sum_{a,b}M_{ab} H_{ab}} }_{GUE} = e^{\frac{1}{2N} \sum_{a,b} M_{ab}  M_{ba} }\;,
\]
hence, using:
\[
 ( \bar \phi_a \phi_b + \bar\chi_a \chi_b )( \bar \phi_b \phi_a + \bar\chi_b \chi_a ) = \left( \bar \phi\cdot \phi \right)^2 -  \left( \bar \chi \cdot \chi \right)^2 +
 2 \left( \bar \chi \cdot \phi \right)  \left( \bar \phi \cdot \chi \right) \;,
\]
the expectation of the resolvent becomes:
\[
\omega_N(z) =  \int   [d\bar \chi d\chi]  [d \bar \phi  d\phi] \; \frac{1}{N} \left( \bar \phi \cdot \phi \right)    e^{  - z \bar \phi \cdot \phi - z \bar\chi \cdot \chi  -
 \frac{1}{2N} \left[ \left( \bar \phi\cdot \phi \right)^2 -  \left( \bar \chi \cdot \chi \right)^2 + 2 \left( \bar \chi \cdot \phi \right)  \left( \bar \phi \cdot \chi \right) \right] } \;.
\]

The key idea is that the resolvent, which for now is expressed as an integral over $2N$ commuting and $2N$ anti commuting variables can be expressed as an integral over only $2$ commuting 
variables. The fist step consists in using the Hubbard-Stratonovich transformation:
\begin{align*}
& e^{-\frac{1}{2N} \left( \bar \phi\cdot \phi \right)^2    } = \int [dA] \;  e^{-\frac{N}{2} A^2 + \imath  A \left( \bar \phi\cdot \phi \right)  } \;, \qquad 
 e^{ \frac{1}{2N}  \left( \bar \chi \cdot \chi \right)^2  } = \int [dD] \;  e^{-\frac{N}{2} D^2 +    D \left( \bar \chi\cdot \chi \right)  } \;,
 \crcr
& e^{ - \frac{1}{N}  \left( \bar \chi \cdot \phi \right)  \left( \bar \phi \cdot \chi \right) } = 
\int[d B dC] e^{- N B C +  \left( \bar \chi \cdot \phi \right)  C -     B \left( \bar \phi \cdot \chi \right)  } \;,
\end{align*}
where $[dA] \equiv \frac{\sqrt{N}dA}{\sqrt{ 2\pi} }$, $[dD] \equiv \frac{\sqrt{N}dD}{\sqrt{ 2\pi } }$ and $ \int [dBdC] \equiv \frac{1}{N} \partial_B\partial_C $
such that all the Gaussian integrals are normalized.
The crucial point is that $A,B,C,D$ are now only one dimensional (commuting and anti commuting) variables.
The expectation of the resolvent then becomes:
\begin{align*}
\omega_N = & \int [d\bar \chi d\chi]  [d \bar \phi  d\phi]  [dA][dD] [dBdC]  \;\;  \;
\frac{\left( \bar \phi \cdot \phi \right) }{N}   \times  \crcr
&  \qquad \times
\exp{ \left\{ -\frac{N}{2}   A^2  - \frac{N}{2} D^2 - N BC - 
\begin{pmatrix}
 \bar \phi & \bar \chi 
\end{pmatrix}
\left[ 
\begin{pmatrix}
  z-  \imath     A   &      B \\ 
  -       C  & z -       D
\end{pmatrix} \otimes I_N 
\right]
\begin{pmatrix}
   \phi  \\   \chi 
\end{pmatrix}
\right\} }
\; ,
\end{align*}
where $I_N$ is the identity $N\times N$ matrix. The integral over $\bar \phi,\phi, \bar \chi, \chi$ is now factored and can be directly evaluated:
 \[
\omega_N(z) =  \int [dA][dD] [dBdC]  e^{-\frac{N}{2} A^2  - \frac{N}{2} D^2  - N BC  }
   \frac{1}{N} \left[ \frac{1}{\imath   } \frac{\partial}{\partial A} \right] 
\left(  \frac{  \left( z   -    D \right) }{   \left[ z -\imath   A  +   B  \left( z -    D \right)^{-1}  C \right] } \right)^N  \; .
 \]
Expanding the polynomial in $B$ and $C$, we obtain:
\[
 \omega_N =  \int [dA][dD][dBdC]  e^{-\frac{N}{2} A^2  - \frac{N}{2} D^2 -NBC } 
 \frac{1}{N} \left[ \frac{ 1 }{  i  } \frac{\partial}{\partial A} \right] 
 \left[  
\left( \frac{ z    -   D  }{  z -\imath     A } \right)^N 
\left( 1 -  \frac{ N  }{ \left( z -  \imath    A \right)  \left(z    -    D   \right)  } BC   \right) 
\right] \;,
\]
and computing the derivative with respect to $A$ we get:
\[
\omega_N =   \int [dA][dD] [dBdC] \; e^{-\frac{N}{2} A^2  - \frac{N}{2} D^2 -NBC  }  
  \left[  \frac{  \left(  z   -    D  \right)^N  }{ \left( z - \imath    A \right)^{N+1}  }
      - BC   (N+1)    \frac{  \left( z   -    D  \right)^{N-1}  } { \left( z -   \imath  A \right)^{N+2} }
  \right] \;.
\]
Finally, evaluating  Grassmann integral (taking into account its normalization) we have:
\[
\boxed{ \omega_N(z)  = \int [dA][dD] \;  e^{-\frac{N}{2} A^2  - \frac{N}{2} D^2  }  
  \left[  \frac{  \left(  z    -    D  \right)^N  }{ \left( z - \imath     A \right)^{N+1}  }
      + \frac{N+1}{N}  \frac{  \left(  z    -    D  \right)^{N-1}  } { \left( z -   \imath   A \right)^{N+2} }
  \right] \;. }
\]

\subsubsection{Inverse Laplace transform}

In order to compute the expectation of the Wilson loop, all that remains to be done is to compute the inverse Laplace transform of the resolvent:
\begin{align*}
 I(t,N) & = \int_{\gamma -\imath \infty}^{\gamma + \imath \infty} 
  \frac{ds}{2\pi \imath} \; e^{ts}  \omega_N \left( s \right) = \crcr
   & =\int_{\gamma -\imath \infty}^{\gamma + \imath \infty} 
  \frac{ ds }{2\pi \imath } \; e^{ts}   \int [dA][dD] \;  e^{-\frac{N}{2} A^2  - \frac{N}{2} D^2  }  
    \left[  \frac{  \left(  s    -  D  \right)^N  }{ \left( s - \imath    A \right)^{N+1}  }
      + \frac{N+1}{N} \frac{  \left( s   -    D  \right)^{N-1}  } { \left( s -   \imath   A \right)^{N+2} }
  \right] \;,
\end{align*}
with $\gamma$ a positive real number. The integrand has a pole on the imaginary axis in $s_0=\imath A$. Closing the contour by a semicircle at infinity 
in the left half complex plane and then shrinking to a circle $C$ of finite radius around $\imath A$, we obtain:
\begin{align*}
I(t,N) & =   \int [dAdD] \; e^{-\frac{N}{2}A^2 - \frac{N}{2}D^2} \int_{C} \frac{ds}{2\pi \imath }e^{itA} e^{t(s-iA)} \times \crcr
   & \qquad \qquad \times  \left[ 
   \frac{  \left(  s -\imath A + \imath A   -  D  \right)^N  }{ \left( s - \imath    A \right)^{N+1}  }
      + \frac{N+1}{N} \frac{  \left( s  -\imath A + \imath A     -    D  \right)^{N-1}  } { \left( s -   \imath   A \right)^{N+2} }
\right] \crcr
& = \int [dAdD] \;e^{-\frac{N}{2}A^2 - \frac{N}{2}D^2} \int_{C} \frac{ds}{2\pi \imath }e^{itA}  \times \crcr
&\qquad \qquad  \times \bigg[ \frac{1}{ \left( s - \imath    A \right)^{N+1} } \sum_{n\ge 0} \frac{t^n}{n!} \left(s-\imath A \right)^n \sum_{q=0}^{N} \binom{N}{q} (s-\imath A)^q (\imath A -D)^{N-q} +  \crcr
& \qquad \qquad \qquad  +  \frac{1}{ \left( s - \imath    A \right)^{N+2} }
   \sum_{n\ge 0} \frac{t^n}{n!} \left(s-\imath A \right)^n \frac{N+1}{N}\sum_{q=0}^{N-1} \binom{N-1}{q} (s-\imath A)^q (\imath A -D)^{N-1-q} 
\bigg] \;.
\end{align*}
The integral over $s$ can now be computed using the residue theorem:
\[I(t,N) =  \int [dAdD] \; e^{-\frac{N}{2}A^2 - \frac{N}{2}D^2}  e^{itA} \left[ \sum_{n=0}^N \frac{t^n}{n!} \binom{N}{N-n}  (\imath A -D)^{n}
+ \sum_{n=2}^{N+1} \frac{t^n}{n!} \frac{N+1}{N} \binom{N-1}{N+1-n} (\imath A -D)^{n-2}  \right] \;.
\]
Changing variables to $D = H + \imath A$ we get:
\begin{align*}
I(t,N) &  =     \int[dA] [dH] \;  e^{-\frac{N}{2}A^2 -\frac{N}{2}( H + \imath A )^2 + \imath t A}   \times \crcr
 & \qquad \qquad \times \left[ \sum_{n=0}^N \frac{t^n}{n!} \binom{N}{N-n} (-H)^{n}
+ \sum_{n=2}^{N+1} \frac{t^n}{n!} \frac{N+1}{N} \binom{N-1}{N+1-n}(-H)^{n-2}  \right] = \crcr
& = \int[dA] [dH] \; e^{ -\frac{N}{2} H^2  -  \imath NH A + \imath t A}  
 \left[ \sum_{n=0}^N \frac{t^n}{n!} \binom{N}{N-n} (- H ) ^{n}
+ \sum_{n=2}^{N+1} \frac{t^n}{n!} \frac{N+1}{N} \binom{N-1}{N+1-n} (-H)^{n-2}  \right] \;,
\end{align*}
which is (reinstating the normalizations of the measures on $A$ and $H$):
\begin{align*}
I(t,N) &  = \int \frac{N dAdH}{2\pi } \; e^{ - \frac{N}{2}H^2  -  \imath AN \left( H - \frac{t}{N}  \right) } \times \crcr
& \qquad \qquad \times  \left[ \sum_{n=0}^N \frac{t^n}{n!} \binom{N}{N-n} (-H)^{n}
+ \sum_{n=2}^{N+1} \frac{t^n}{n!} \frac{N+1}{N} \binom{N-1}{N+1-n}(-H)^{n-2}  \right] =\crcr
& = \int dH  \;  \delta\left( H - \frac{t}{N} \right) \;  e^{-\frac{N}{2}H^2}  \times \crcr
   & \qquad \qquad \times \left[ \sum_{n=0}^N \frac{t^n}{n!} \binom{N}{N-n} (-H)^{n}
+ \sum_{n=2}^{N+1} \frac{t^n}{n!} \frac{N+1}{N} \binom{N-1}{N+1-n} (-H)^{n-2}  \right] = \crcr
& =  e^{ -\frac{t^2}{2N} }   \left[ \sum_{n=0}^N \frac{t^n}{n!} \binom{N}{N-n} \left( -\frac{t}{N} \right)^{n}
+ \sum_{n=2}^{N+1} \frac{t^n}{n!} \frac{N+1}{N} \binom{N-1}{N+1-n}\left( -\frac{t}{N} \right)^{n-2}  \right] = \crcr
& =  e^{ -\frac{t^2}{2N} }  \left[ \sum_{q=0}^N (-t^2)^q\frac{1 }{ q!} \frac{1}{N^q} \binom{N}{q} 
   +\sum_{q=1}^N  (-t^2)^q  \frac{ -1}{(q+1)! } \frac{N+1}{N} \binom{N-1}{N-q} \frac{1}{N^{q-1}}     
\right] \;,
\end{align*}
where in the second sum $q = n-1$. Separating the term with $q=0$ and observing that the term with $q=N$ is zero we obtain:
\[
 I(t,N)   = e^{ -\frac{t^2}{2N} }  \left\{ 1 + \sum_{q=1}^{N-1} (-t^2)^q \left[\frac{1 }{ q!} \frac{1}{N^q} \binom{N}{q}  -  \frac{ 1}{(q+1)! } \frac{N+1}{N} \binom{N-1}{N-q} \frac{1}{N^{q-1}}  \right] \right\} \; ,
\]
and Theorem \ref{thm:WL} follows as:
\begin{align*}
&  \frac{1 }{ q!} \frac{1}{N^q} \binom{N}{q}  -  \frac{ 1}{(q+1)! } \frac{N+1}{N} \binom{N-1}{N-q} \frac{1}{N^{q-1}}  =
 \frac{ N!  }{ q! N^q q! (N-q)! } - \frac{ (N+1) (N-1)! }{ (q+1)! N (N-q)! (q-1)! N^{q-1} } = \crcr
& = \frac{ (N-1)!}{ q! N^q (q-1)! (N-q)! } \left[ \frac{N}{q} - \frac{N+1}{ q+1 }\right] = \frac{ (N-1)!}{ q! N^q (q-1)! (N-q)! } \; \frac{ (N-q) }{q(q+1)} 
 = \frac{ (N-1)! }{ q! (q+1)! N^q (N-1-q)  } =\crcr
 & = \frac{1}{q!} \frac{1}{N^{q+1}} \binom{N}{q+1} \; .
\end{align*}

\qed

\section{The spectral density}\label{sec:spectden}

The spectral density in the GUE \cite{DiFrancesco:1993nw} is the Fourier transform in the sense of distributions of the expectation of the Wilson loop:
\[
 I(t,N)  = \Braket{\frac{1}{N} \Tr\left[ e^{\imath t H } \right] }_{GUE} = \int d\lambda \; \rho_N(\lambda) e^{\imath t \lambda}  \Rightarrow \rho_{N}(\lambda)  = \frac{1}{2\pi} \int_{-\infty}^{\infty} dt \;   e^{-\imath \lambda t}  I(t,N) \; .
\]
This formula can also be recovered by noting that the spectral density is the discontinuity of the resolvent:
\begin{align*}
& \rho_{N}(\lambda) = \frac{1}{2\pi \imath} \lim_{\epsilon\to 0} 
 \left[ \frac{1}{N} \Braket{ \frac{1}{H - \lambda - \imath \epsilon} }_{GUE} - \frac{1}{N} \Braket{ \frac{1}{H - \lambda + \imath \epsilon}  }_{GUE} \right]  = \crcr
& = \frac{1}{2\pi} \lim_{\epsilon\to 0} \left[ \omega_N(\epsilon+\imath \lambda) - \omega_N(- \epsilon+\imath \lambda)  \right]  
  = \frac{1}{2\pi} \lim_{\epsilon\to 0}   \int_{0}^{\infty} dt \;  e^{-\epsilon t} \left[  e^{-\imath \lambda t} +  e^{ \imath \lambda t} \right]I(t,N) =\crcr
& =\frac{1}{2\pi} \int_{-\infty}^{\infty} dt \;   e^{-\imath \lambda t}  I(t,N) \; .
\end{align*}
Using Theorem \ref{thm:WL}, the spectral density at finite $N$ writes:
\[
\boxed{  \rho_N(\lambda)   = \frac{1}{2\pi} \sum_{q=0}^{N-1} \frac{1}{N^{q+1}} \binom{N}{q+1} \frac{1}{q!}
    \int_{-\infty}^{\infty} d t\;  e^{\imath t \lambda  -\frac{t^2}{2N} }  (-t^2)^q  
  =   \sum_{q=0}^{N-1}\frac{1}{N^{q+1}} \binom{N}{q+1} \frac{1}{q!}  \left( \frac{\partial}{ \partial \lambda} \right)^{2q}
     \left[ \sqrt{\frac{ N}{2\pi} }  e^{-\frac{N\lambda^2}{2}}\right] \;, }
\]
which can be further written in terms of Hermite functions.

Albeit this formula is exact, it is somewhat difficult to handle. In  particular, the large $N$ limit and the $1/N$ expansion are not obvious. In order to study the spectral density 
at finite $N$ in more detail, let us compute it moments (expectations of monomials in $\lambda$). The odd moments are zero, while the even ones are:
\begin{align}\label{eq:momecomplet}
   \Braket{ \frac{ \Tr \left[H^{2l} \right] }{N} }_{GUE} & =\Braket{\lambda^{2l}} = \int d\lambda \;\; \rho_N(\lambda) \lambda^{2l}  =\crcr
   & = \sum_{q_2=0}^{N-1}\frac{1}{N^{q_2+1}} \binom{N}{q_2+1} \frac{1}{q_2!}  \int d\lambda \; 
  \left\{ \left( \frac{\partial}{ \partial \lambda} \right)^{2q_2}
     \left[ \sqrt{\frac{ N}{2\pi} }  e^{-\frac{N\lambda^2}{2}}\right] \right\} \lambda^{2l}  = \crcr
     & = \sum_{q_2=0}^{\min\{l,N-1\}}\frac{1}{N^{q_2+1}} \binom{N}{q_2+1} \frac{1}{q_2!}  \int d\lambda \;    
     \sqrt{\frac{ N}{2\pi} }  e^{-\frac{N\lambda^2}{2}}  \frac{(2l)!}{(2l-2q_2)!} \lambda^{2l-2q_2}  = \crcr
     & = \sum_{q_2=0}^{\min\{l,N-1\}}\frac{1}{N^{q_2+1}} \binom{N}{q_2+1} \frac{1}{q_2!}   \frac{(2l)!} {2^{l-q_2} (l-q_2)! } \frac{1}{N^{l-q_2}} \; .
\end{align}
In this form one can for instance recover the large $N$ limit of the moments:
\[
 \lim_{N\to \infty}  \Braket{ \frac{ \Tr \left[H^{2l} \right] }{N} }_{GUE} = \frac{(2l)!}{l!(l+1)!} \;,
\]
reproducing, as expected, the moments of the Wigner semicircle distribution  $\rho_{\infty}(\lambda) = \frac{1}{2\pi} \sqrt{4-\lambda^2}$.
However, the $1/N$ expansion is somewhat involved, as the expansion of the binomial coefficient in $1/N$ is alternated. For instance 
all the odd powers of $1/N$ in Eq.~\eqref{eq:momecomplet} cancel out, which seems somewhat mysterious in the equation as written.

In order to recover the $1/N$ expansion of 
the moments as a series in $1/N$ with positive coefficients we use the following identity:
\[
\int_{C_R} \frac{dw}{2\pi \imath} \;\; \frac{1}{2} \left( \frac{w +\frac{1}{N}}{ w-\frac{1}{N} } \right)^N \left(  w-\frac{1}{N} \right)^{q_2}
  = \begin{cases}
    \frac{2^{q_2}}{ N^{q_2+1} } \binom{N}{q_2+1} \; ,  \qquad  & q_2 \le N-1 \\
    0 \; ,  \qquad  & q_2 \ge  N
     \end{cases} \;,
\]
where $C_R$ is a circle of radius $R>1$ in the complex plane centered at $1/N$. While for now this  representation is just a convenient trick, it is
natural in view of Section \ref{ssec:HZ}.
The even moments at finite $N$ admit the following integral representation:
\begin{align*}
  \Braket{ \frac{ \Tr \left[H^{2l} \right] }{N} }_{GUE} & = \frac{(2l)!}{2^ll!} 
  \int_{C_R} \frac{dw}{2\pi \imath} \;\; \frac{1}{2} \left( \frac{w +\frac{1}{N}}{ w-\frac{1}{N} } \right)^N
  \sum_{q_2=0}^l \binom{l}{q_2} \left(  w-\frac{1}{N} \right)^{q_2} \frac{1}{N^{l-q_2}} = \crcr
    & = \frac{(2l)!}{2^ll!} 
  \int_{C_R} \frac{dw}{2\pi \imath} \;\; \frac{1}{2} \left( \frac{ 1 +\frac{1}{wN}}{ 1-\frac{1}{wN} } \right)^N w^l \;.
\end{align*}
This integral can now be expanded in $1/N$ using:
\[
  \left( \frac{1+\frac{1}{wN}}{1-\frac{1}{wN}} \right)^N   = e^{N \left[ \ln ( 1+\frac{1}{wN} )- \ln ( 1-\frac{1}{wN} ) \right]}
 = e^{ \sum_{q\ge 0} \frac{1}{N^{2q}} \frac{2}{2q+1} \frac{1}{w^{2q+1}} } \;,
\]
to obtain:
\begin{align}\label{eq:expmome}
 \Braket{ \frac{ \Tr \left[H^{2l} \right] }{N} }_{GUE} & = \frac{(2l)!}{2^ll!}  \sum_{g\ge 0} \frac{1}{N^{2g}} 
  \int_{C_R} \frac{dw}{2\pi \imath} \;\;  \frac{1}{2} 
 \left[   \sum_{k_q\ge 0 }^{\sum qk_q = g} \prod_{q\ge 0}  \frac{1}{   k_q!}  \left( \frac{2}{(2q+1)w^{2q+1}} \right)^{k_q} \right] w^l =\crcr
 & =  \frac{(2l)!}{l!}  \sum_{g\ge 0} \frac{ 1 }{2^{2g} N^{2g}} 
 \left[   
  \sum_{k_q\ge 0 }^{\genfrac{}{}{0pt}{}{ \sum qk_q = g }{ \sum_q k_q = l-2g+1 } } \prod_{q\ge 0}  \frac{1}{   k_q!}    
  \frac{1}{ (2q+1)^{k_q} }   
 \right] \;. 
\end{align}
 
\section{Connection with enumerative combinatorics}\label{sec:combi}
 
In this section we derive Theorem \ref{thm:WL} using enumerative combinatorics techniques and then use it 
to find a novel proof of the Harer-Zagier series formula.
 
\subsection{A BEST theorem for undirected multi-graphs}\label{ssec:BEST}

A standard manipulation in field theory consists in representing Gaussian integrals as differential operators. In particular the expectation
of a monomial in $H$ in the GUE is:
\[
   \Braket{ \frac{ \Tr \left[H^{2l} \right] }{N} }_{GUE}   = \left[ e^{\frac{1}{2 N} \sum_{a,b=1}^N \frac{\partial}{\partial H_{ab}}   \frac{\partial}{\partial H_{ba}} } 
   \frac{ \Tr \left[H^{2l} \right] }{N}\right]_{H=0} =
   \left[ \frac{1}{l!} \frac{1}{  N^{l+1}}
   \left(\frac{1}{2} \sum_{a,b=1}^N \frac{\partial}{\partial H_{ab}}   \frac{\partial}{\partial H_{ba}} \right)^l 
   \Tr \left[H^{2l} \right] 
   \right]_{H=0} \;.
\]
Combining this with Eq.~\eqref{eq:momecomplet}  yields:
\begin{align}\label{eq:initial}
  \frac{1}{l!}\left[   \left( \frac{1}{2}\sum_{a,b=1}^N \frac{\partial}{\partial H_{ab}}   \frac{\partial}{\partial H_{ba}} \right)^l 
   \Tr \left[H^{2l} \right] 
   \right]_{H=0} = \sum_{q_2=0}^{\min\{l,N-1\}} \binom{N}{q_2+1}  \; \;   \frac{(2l)!}{2^{l-q_2} q_2! (l-q_2)! } \;.
\end{align}
Both sides of this equation admit a combinatorial interpretation.

\paragraph{The left hand side.}
Let us consider $N$ vertices labeled $1,\dots N$. An \emph{undirected multi graph} over a subset of the vertices 
$\{1,\dots N\}$ is identified by the multiplicities $l^G_{ab}\ge 1$ of the edges $\{a,b\}=\{b,a\}$:
 \[
  G = \bigg{\{} \{a,b\}^{l^G_{ab}} \; | \;  \{ a,b \} \subset \{1,\dots N\}  \bigg\} \;.
 \]
Both multiple edges and self loops are allowed. The multiplicities of the edges can be completed to a list $l^G_{pq}$ for any pair 
$\{p,q\} \subset\{1,\dots N\}$  by setting $l^G_{pq}=0$ if there is no edge in $G$ connecting $p$ and $q$.
We denote $E(G)$ the (multi) set of the edges of $G$ and $V(G) \subset \{1,\dots N\}$ the set of vertices of $G$.

Let us denote $ \frac{\partial}{\partial H_{ab}}  \equiv \partial_{ab} $.
The differential operator in the left hand side of Eq.~\eqref{eq:initial} can be expanded in multi graphs $G$:
\begin{align*}
 \frac{1}{l!} \left(\frac{1}{2}  \sum_{a,b=1}^N\partial_{ab}  \partial_{ba}  \right)^{l} & = \left(\sum_{a=1}^N \frac{1}{2} (\partial_{aa})^2 + \sum_{1\le a<b\le N} \partial_{ab} \partial_{ba} \right)^l  =
 \crcr
 & =\sum_{G}^{|E(G)|=l}    \left(  \prod_{a < b}  \frac{1}{l^G_{ab}!} \left( \partial_{ab} \partial_{ba} \right)^{l^G_{ab}}  \right)  
 \left( \prod_{a} \frac{1}{l^G_{aa}!} \left( \frac{1}{2} (\partial_{aa})^2  \right)^{l^G_{aa}} \right) \;.
\end{align*}
Let us reorganize the sum by the number $|V(G)|= q_2+1$ of vertices of $G$:
\[
 \frac{1}{l!} \left(\frac{1}{2}  \sum_{a,b=1}^N\partial_{ab}  \partial_{ba}  \right)^{l}  
 = \sum_{q_2 = 0}^{\min\{l,N-1\}} \sum_{G}^{\genfrac{}{}{0pt}{}{|V(G)| = q_2+1}{|E(G)| = l }} 
     \left(  \prod_{a < b}  \frac{1}{l^G_{ab}!} \left( \partial_{ab} \partial_{ba} \right)^{l^G_{ab}}  \right) 
 \left( \prod_{a} \frac{1}{l^G_{aa}!} \left( \frac{1}{2} (\partial_{aa})^2  \right)^{l^G_{aa}} \right) \; .
 \]

Let us associate to any undirected multi graph $G$ the \emph{directed multi graph} ${\rm di}(G)$ obtained by splitting all the undirected edges $\{a,b\}$ of $G$ into a pair of directed arcs $(a,b)$ and $(b,a)$. 
The number of arcs of ${\rm di}(G) $ is twice the number of edges of $G$. Self loops $\{a,a\}$ in $G$ are also split in $ {\rm di}(G)$ into pairs of directed arcs $(a,a)$ and $(a,a)$.
The arcs of $ {\rm di}(G)$ are canonically paired  into pairs corresponding to the edges of $G$.

An \emph{Eulerian cycle} in ${\rm di}(G)$ is a cycle which visits every arc of ${\rm di}(G)$  exactly once respecting its orientation.
Such a cycle can be rooted by fixing its first step to be a certain fixed arc.
We denote $N^{{\rm E- di}}(G)$ the number of Eulerian cycles of ${\rm di}(G)$.

 \begin{proposition}
 For any multi graph $G$ with $l$ edges we have:
\[
   \left(  \prod_{a < b}  \frac{1}{l^G_{ab}!} \left( \partial_{ab} \partial_{ba} \right)^{l^G_{ab}}  \right) 
 \left( \prod_{a} \frac{1}{l^G_{aa}!} \left( \frac{1}{2} (\partial_{aa})^2  \right)^{l^G_{aa}} \right) \Tr[H^{2l}]    =  N^{\rm E-di}(G) \;.
\]
\end{proposition}
 \begin{proof} Consider a term with fixed indices in the trace. The first $H$ in this term has a first index, say $a$. 
 Some derivative must act on it to go to its second index, say $b$ which is identical with the first index on the second $H$
and some derivative must act on the second $H$ and so on. Hence any two indices in $\Tr[H^n]$ are connected by a path of 
oriented arcs $(a,b)$ and every edge $\{a,b\}$ in $G$ is encountered exactly twice in the derivative operator, once corresponding to the arc $(a,b)$ and
once corresponding to the arc $(b,a)$.

\end{proof}

Observe that Eulerian cycles in ${\rm di}(G)$ exist if and only if $G$ is connected.

In particular, the result of acting with the derivatives corresponding to a multi graph on a trace does not depend on the labels of its vertices: given a connected graph with $q_2+1$ vertices, there are
$\binom{N}{q_2+1}$ labeling of the vertices which lead to the same result. The left hand side in Eq.~\eqref{eq:initial} is then:
\[
 \sum_{q_2 = 0}^{\min\{l,N-1\}} \binom{N}{q_2+1} \sum_{G\; {\rm connected}}^{\genfrac{}{}{0pt}{}{V(G) =\{1,\dots   q_2+1\} }{|E(G)| = l }}  N_r^{\rm E-di}(G) \;,
\]
where the sum runs over multi graphs whose vertices have fixed labels $\{1,\dots q_2+1 \}$.

\paragraph{The right hand side.}
The right hand side of Eq.~\eqref{eq:initial}:
\[
 \sum_{q_2=0}^{\min\{l,N-1\}} \binom{N}{q_2+1}  \; \;   \frac{(2l)!}{2^{l-q_2} q_2! (l-q_2)! } \;,
\]
also has a combinatorial interpretation. Indeed, $ \frac{(2q_2)!}{q_2!}$
is the number of plane trees with vertices labels $\{1,\dots q_2\}$, and:
\[
\frac{1}{2^{l-q_2} (l-q_2)! }  \frac{(2l)!}{(2q_2)!}\;,
\]
counts the number of ways to  decorate a plane tree by $(l-q_2)$ excess loop edges to obtain a \emph{combinatorial map} $M$, hence:
\[
 \frac{(2l)!}{2^{l-q_2} q_2! (l-q_2)! }  = \sum_{M \; {\rm connected}}^{\genfrac{}{}{0pt}{}{V(M) = \{1,\dots q_2+1\}}{|E(M)|=l}} 1\;.
\]

A rooted combinatorial map is a combinatorial map with a marked edge and an orientation chosen for this edge.
Any combinatorial map $M$ has a canonically associated undirected multi graph ${\rm Gr}(M)$ 
with the same vertices and the same edges as $M$: ${\rm Gr}(M)$ ignores the order of the half edges  around the vertices of $M$.

It follows that Eq.~\eqref{eq:initial}, which we proved by a direct computation, is implied by the following (stronger) combinatorial result.

\begin{proposition}
Let $G$ be an undirected multi graph with vertices labeled $\{1,\dots q_2+1\}$. There is a one to one correspondence 
between the Eulerian cycles in ${\rm di} \left(G\right)$ rooted at $(a,b)$ and the combinatorial maps 
$M$ with ${\rm Gr}(M) = G$ rooted at $(a,b)$ with a chosen (plane) spanning tree $T$.
\end{proposition}

\begin{proof}

 The proof of this theorem is very close to the proof of the BEST theorem \cite{aardenne1951circuits,tutte1941unicursal} in combinatorics.
 If $G$ has multiple edges and self loops, one first labels the multiple edges and self loops and then proceeds.  An example of the bijection described 
 below is presented in Fig.~ \ref{fig:Example}.

 \begin{figure}[htb]
\begin{center}
\includegraphics[width=12cm]{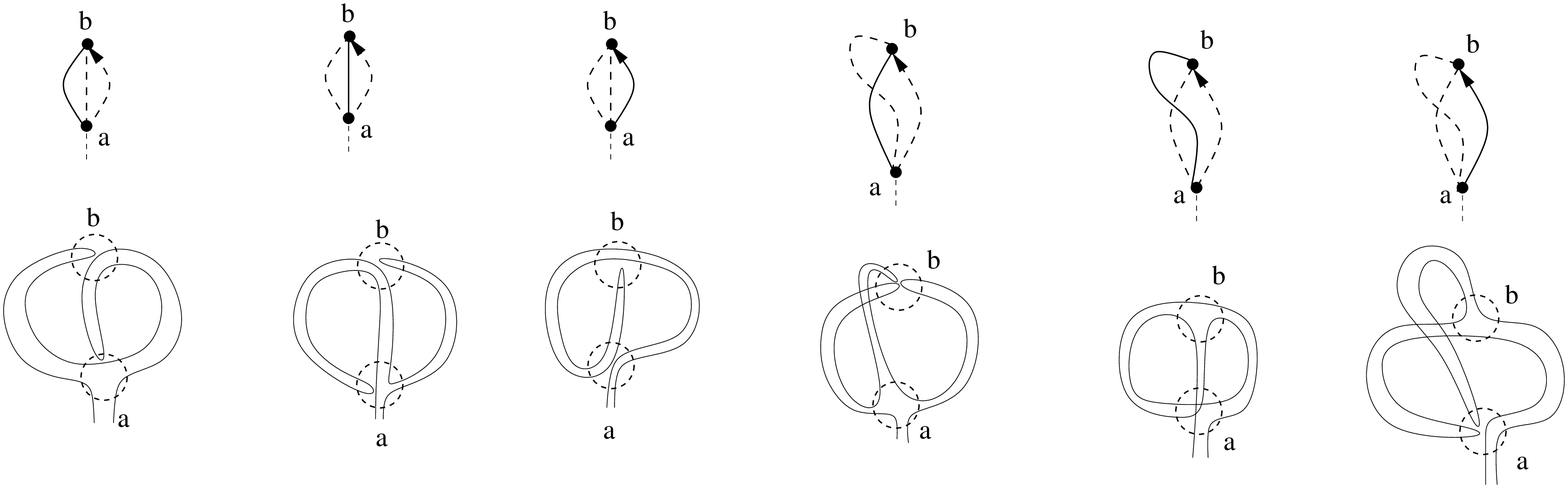}  
\caption{Maps $M$ with ${\rm Gr}(M)=G$ with chosen trees $T$ and Eulerian paths in ${\rm di}(G)$. The tree edges are solid, the loop edges are dashed and the 
root edge is indicated by an arrow.\label{fig:Example}} 
\end{center}
\end{figure}  
 
{\bf From a map $M$ (with ${\rm Gr}(M)=G$) and a tree $T$ to an Eulerian cycle in ${\rm di}(G)$.} 
A clockwise orientation of the faces of the map induces opposite orientations on the sides of the edges. Around any vertex, the sides of the edges 
are alternatively outgoing and incoming \emph{i.e.}, turning counterclockwise around a vertex, the outgoing side of the edge $e$ is followed by the incoming side of the edge $e$ which is followed by the outgoing side of the successor of $e$ and so on.
Stepping along the sides of the edges in $M$ is equivalent to stepping along the corresponding arcs in ${\rm di}(G)$. 
Start by stepping along the side $(a,b)$ of the root edge, from $a$ to $b$. At all the subsequent steps, leave the current vertex $v$ on the first available outgoing side of an edge 
(\emph{i.e.} not yet used in the cycle) counted counterclockwise starting from the incoming side of the unique edge in $T$  hooked to $v$ which goes towards the root vertex $a$ 
(or starting from the incoming side of the root edge $(a,b)$ itself if $v=a$).
It is easy to show that the cycle thus obtained is Eulerian in ${\rm Gr}(M)$ by induction starting from the root vertex.

 {\bf From a cycle in ${\rm di}(G)$ to $M$ and $T$.} List the arcs in the cycle in the order they are encountered:
  \[
   (a,b) (b,c) \dots \;.
  \]
  The edges of $M$ are the canonical pairs of arcs of ${\rm di}(G)$.
  Build the vertices of $M$ by adding the edges around a vertex $v$ in the order in which the outgoing arcs $(v,\cdot)$ are encountered
  in the list. Declare the last edge in the list which exits $v$ as the tree edge hooked to $v$ which goes towards the root.

  \end{proof}

 \subsection{The Harer-Zagier series formula}\label{ssec:HZ}
 
One can attempt to compute the expectation of the Wilson loop directly. Expanding the exponential in Eq.~\eqref{eq:WilsonLoop} and taking into account that a Gaussian integral with an odd number of insertions is $0$ gives:
\begin{equation}\label{eq:first}
  I(t,N) = \sum_{l\ge 0} \frac{ (-t^{2})^{l} }{(2l)!}   \; \int [dH]  \; e^{-\frac{N}{2} \Tr H^2}
    \; \frac{1}{N  } \Tr(H^{2l})  \; .
\end{equation}

 This Gaussian integral can be evaluated \cite{DiFrancesco:1993nw} as a sum over Feynman graphs which in the case of the GUE are rooted combinatorial maps.
 Each map has only one vertex of coordination $2 l$ (hence $l$ edges) and
 $F = 1+l -2g$ faces, where $g$ is the genus of the map.  Such maps are sometimes called rooted \emph{rosettes}. 
 Each rosette contributes to the evaluation of Eq.~\eqref{eq:first} a term $N^{-2g}$
therefore, denoting $C_g(l)$ the number of rooted rosettes with a vertex of coordination $2l$ and genus $g$, we have:
\[
 I(t,N)   =
 \sum_{l\ge 0} \frac{(-t^{2})^{l} }{(2l)!}   \sum_{g\ge 0} N^{-2g} C_g(l)    = 1+ \sum_{g\ge 0, l>0} \frac{(-t^{2})^{l}  }{2^l l!}  \;  \frac{ C_g(l) }{(2l-1)!!}\frac{1}{N^{2g}} \;.
\]

Concerning the numbers $C_g(l)$, it is well known that $\sum_{g\ge 0} C_g(l) = (2l-1)!!$ and $C_0(l) = \frac{1}{l+1} \binom{2l}{l}$
which in particular implies that for all $N\ge 1$ and for all complex $t$, $|I(t,N)| \le e^{\frac{1}{2}| t|^2} $ (a much weaker bound than the one we derived starting from the exact formula).
Eq.~\eqref{eq:expmome} yields an exact formula for $C_g(l)$:
\[
 C_g(l) = \frac{(2l)!}{ l!  2^{2g}} \sum_{k_q\ge 0 }^{ \genfrac{}{}{0pt}{}{\sum qk_q = g} {\sum k_q = l-2g+1}} \prod_{q\ge 0}  \frac{ 1}{   k_q! (2q+1)^{k_q}} \;,
\]
which coincides with Proposition 6 of  \cite{DBLP:journals/jct/ChapuyFF13} upon summing over $k_0$.
Let us denote $f(x,N)$ the generating function of the numbers $C_g(l)$:
\[
 f(x,N) = \sum_{p > 0,\sum g \ge 0}\frac{C_g(p)}{(2p-1)!!} x^{p+1}  \frac{1}{N^{2g}} \;.
\]
The main remark is that $I(\imath \sqrt{2u},N) -1$ is the Borel transform of $ \frac{1}{x } f(x,N)$, therefore: 
  \[
  \frac{1}{x} f(x,N)  =  \frac{1}{x} \int_0^{\infty} du \; e^{-\frac{u}{x} } \left( I(\imath \sqrt{2u}) -1\right) \;,
 \]
and using Theorem \ref{thm:WL} we obtain:
 \begin{align*}
f(x,N) & = \int_0^{\infty} du \; e^{-\frac{u}{x} } \left[ e^{\frac{u}{N}}  \sum_{q=0}^{N-1}  \frac{1}{N^{q+1}} \binom{N}{q+1} \frac{ 2^q u^q } {q!}   -1\right] 
    =  \sum_{q=0}^{N-1}  \frac{1}{N^{q+1}} \binom{N}{q+1} \frac{2^q}{q!} \frac{1}{\left( \frac{1}{x}-\frac{1}{N} \right)^{q+1}} q! - x\crcr
   & = \frac{1}{2} \sum_{q=0}^{N-1} \binom{N}{q+1} \left(  \frac{2x}{N} \right)^{q+1} \frac{1}{ \left( 1 - \frac{x}{N} \right)^{q+1} } -x  
   = \frac{1}{2} \left(  1+  \frac{ \frac{2x}{N}}{1 -\frac{x}{N}}   \right)^{N} - \frac{1}{2} - x \crcr
   & =\frac{1}{2}  \left(   \frac{1+\frac{x}{N} } {1 - \frac{x}{N} } \right)^N -\frac{1}{2} - x\;,
 \end{align*}
 reproducing the Harer-Zagier \cite{harer1986euler} series formula.

\bibliography{/home/razvan/Desktop/lucru/Ongoing/Refs/Refs.bib}{}

\end{document}